%% file: convolution_arxiv.tex
\RequirePackage{fix-cm}
\documentclass[12pt,a4paper,notitlepage]{article}
\pdfoutput=1 
\usepackage[hyphens]{url}
\usepackage{hyperref}

\usepackage{amsmath,mathtools,etoolbox,enumerate,xspace,amsthm,amssymb,amsfonts}
\usepackage{microtype}
\mathtoolsset{mathic,centercolon}

\usepackage{xcolor}
\hypersetup{
  colorlinks,
  linkcolor={red!70!black},
  citecolor={green!50!black},
  urlcolor={blue!50!black},
  pdfauthor={Floris van Doorn},
  unicode,
  pdftitle={Designing a general library for convolutions}}
\usepackage[numbers]{natbib}

\newtheorem{theorem}{Theorem}[section]

\newtheorem{proposition}[theorem]{Proposition}
\newtheorem{lemma}[theorem]{Lemma}
\newtheorem{corollary}[theorem]{Corollary}

\theoremstyle{definition}
\newtheorem{definition}[theorem]{Definition}

\theoremstyle{remark}

\usepackage{ucs}
\usepackage[utf8x]{inputenc}
\usepackage[LGR, T1]{fontenc}
\usepackage{textcomp}
\usepackage{textgreek}
\usepackage{tikz-cd}
\usepackage{mathrsfs}

\usepackage[english]{babel}

\usepackage[capitalize]{cleveref}
\usepackage{xspace}

\usepackage{fontawesome} 
\usepackage{mathtools} 
\usepackage{upgreek} 

\definecolor{keywordcolor}{rgb}{0.7, 0.1, 0.1}   
\definecolor{commentcolor}{rgb}{0.4, 0.4, 0.4}   
\definecolor{symbolcolor}{rgb}{0.0, 0.1, 0.6}    
\definecolor{sortcolor}{rgb}{0.1, 0.5, 0.1}      
\definecolor{errorcolor}{rgb}{1, 0, 0}           
\definecolor{stringcolor}{rgb}{0.5, 0.3, 0.2}    

\usepackage{listings}

\usepackage{flushend}

\input{macros_env.tex}

\title{Designing a general library for convolutions}

\author{Floris van Doorn}
\date{}

\newcommand{\mllink}[1]{%
\href{https://github.com/leanprover-community/mathlib/blob/ec5f9adc44ba5559369652f24b5b1e8231dbc981/src/#1}{\link}}

\newcommand{\clink}[1]{\mllink{analysis/convolution.lean\#L#1}}

\begin{document}
\maketitle

\begin{abstract}
  We will discuss our experiences and design decisions obtained from building a formal library for
  the convolution of two functions.
  Convolution is a fundamental concept with applications throughout mathematics.
  We will focus on the design decisions we made to make the convolution general and easy to use,
  and the incorporation of this development in Lean's mathematical library \textsf{mathlib}.
\end{abstract}

\section{Introduction}%
\label{sec:introduction}

Convolutions are used throughout mathematics.
They are used in probability theory to compute the probability density function
for the sum of two independent random variables,
in image processing to blur images and detect edges,
in differential topology to approximate functions with smooth functions,
and it is used to define distributions as generalized functions.

In the most basic form, the convolution of two functions $f,g : \R \to \R$ is defined to be
$$(f * g) x = \int f(t)g(x-t)dt.$$

However, there are many variants and generalizations of convolutions:
\begin{itemize}
  \item Discrete convolution of functions defined on $\Z$;
  \item Convolution for functions defined on $\R^n$
  \item More generally, convolution for functions defined on a group $G$;
    equipped with a measure $\mu$;
  \item Convolution of a scalar function with a vector-valued function.
\end{itemize}
It is tricky to design a general library for convolution capturing all these variations,
especially if you realize that the various instances of convolution satisfy different properties.
For example, the convolution operator for functions defined on $\R^n$ is commutative,
but this no longer holds on (nonabelian) groups $G$.
And while the convolution of functions on $\R^n$ is smooth under suitable hypotheses,
this statement is nonsensical for the convolution of functions on $\Z$.

In this paper we discuss the design of a convolution library we designed
for the Lean 3 theorem prover~\cite{moura2015lean}. Lean is
developed principally by Leonardo de Moura at Microsoft
Research and implements a version of the calculus of inductive
constructions~\cite{coquand1988inductive}
with quotient types, non-cumulative universes, and proof irrelevance.
We use and incorporate our library into its mathematical library \textsf{mathlib}~\cite{mathlib},
which has more than 250 contributors and almost one million lines of code.
We give a single definition of convolution that captures all aforementioned cases,
and we prove all properties about the convolution
using the minimal assumptions needed for that property.
This flexibility in assumptions for each theorem is greatly aided
by the flexible type-class system implemented in Lean. 

In fact, we generalize the convolution enough
in order to compute the total derivative of the convolution of two multivariate functions,
even in the case where one of them takes values in any Banach space.
To do this we have to define a convolution in the case where both functions take values in
a normed space, and we use a bilinear map to combine the functions.
We could not find the computation
of the total derivative of the convolution anywhere in the literature;
all books that the author of this paper consulted compute only the partial derivatives
of the convolution~\cite{hirsh1976differentialtopology,yosida1995functionalanalysis}.

This formalization is interesting since the convolution has a lot of counterintuitive properties.
For example, the convolution of two smooth functions need not exist.
And even if the convolution exists everywhere, it need not be continuous.
This makes the formalization of these results particularly useful,
since one can formally check which results are true and what assumptions are needed for each result.

To get well-behaved properties of the convolution,
one needs to impose conditions on one or both functions stating that
they decay sufficiently rapidly at $\infty$.
The strongest such condition, on which we will focus,
is the requirement that one of the functions has compact support.
This is a condition that will ensure that the convolution well-behaved,
under a very weak condition of the other function (local integrability).

Other conditions could be imposed on the functions to obtain a well-behaved convolution.
For example, one can assume that both functions in the convolution are rapidly decreasing.
However, such results are typically not more general than the results in this paper,
since a relaxation of the behavior of one function 
has to be counteracted by a stricter condition on the other function.

We prove the following properties about convolution:
\begin{itemize}
\item Algebraic properties: commutativity, associativity and bilinearity;
\item Integrability properties for the convolution of integrable functions;
\item Smoothness properties: continuity, differentiability and being $C^n$;
\item Approximation properties: the convolution of $f$ with a small bump function is close to $f$.
\end{itemize}
The last property shows the power of convolutions:
they give a very general framework to approximate functions with smooth functions,
even if the original function has very weak properties (being locally integrable).
For the smoothness properties we assume that one of the functions has compact support.

When formalizing mathematics, there is often a difference in formalizing a specific theorem
and writing a reusable library that is nice and convenient to use.
In the former one often takes shortcuts when setting up the development
in a manner that is sufficient for proving the relevant theorem.
When designing a reusable library, one wants to focus more on giving very general definitions,
and proving the properties in a great generality, so that the library can be reused in many
areas. One often also spends more time on reliable naming conventions,
documentation and code readability.
These differences are one of the reasons that the reviewing process for \textsf{mathlib}
contributions is lengthy, especially for new users.
The written code has to be tidied,
pass the linter suite of \textsf{mathlib}~\cite{vandoorn2020maintaining}
and satisfy \textsf{mathlib}'s standards for code quality.

This paper is not a formalization of a single big theorem,
but instead describes the design of a reusable library of a tool, to be used in other developments.
Unless specifically mentioned,
all results in this paper have been incorporated into \textsf{mathlib}.
We will link to specific results in \textsf{mathlib}
using the icon \href{https://github.com/leanprover-community/mathlib}{\link} using
static links to the current version of \textsf{mathlib}.

There are many other proof assistants that have developed a substantial analysis library,
most notably Isabelle~\cite{holzl2013analysis}, HOL Light~\cite{harrison2013euclidean}
and Coq~\cite{boldo2015coquelicot}.
However, to my knowledge there are no proof assistants that have formalized the properties of
the convolution of smooth functions.

\section{Preliminaries}%
\label{sec:preliminaries}


\subsection{Type Classes}
\label{sub:typeclasses}
\textsf{mathlib} uses partially bundled
\emph{type-classes}~\cite{wadler1989polymorphism,spitters2011typeclasses} throughout the library.
Almost all results in \textsf{mathlib}
depend on one or more type-classes, and this paper we will see a couple of them.
For example, the type-class
\lstinline{[nontrivially_normed_field 𝕜]}\mllink{analysis/normed/field/basic.lean\#L513}
states that $\kk$ is a field
with a norm $\|{\cdot}\|:\kk\to\R$ satisfying
\begin{itemize}
  \item $d(x,y)\vcentcolon=\|x-y\|$ defines a metric on $\kk$;
  \item $\|xy\|=\|x\|\cdot\|y\|$;
  \item the norm on $\kk$ is nontrivial: there is an element $x\in\kk$ with $\|x\|>1$.
\end{itemize}
Outside code-snippets,
we will use the term \emph{normed field} to refer to nontrivially normed fields.
When applying a lemma that depends on a type-class,
Lean will automatically find these type-class arguments during \emph{type-class synthesis},
using the data provided in a theorem and globally-declared instances.
An example of a globally declared instance is \lstinline{[nontrivially_normed_field ℝ]}.%
\footnote{To be precise, this is a consequence of $\R$ being a densely-ordered field,
but this derivation happens automatically by type-class synthesis.}

If $\kk$ is a normed field, then we can consider all normed spaces over $\kk$.
For technical reasons we split this condition in two parts.

The use the class \lstinline{[normed_add_comm_group E]}\mllink{analysis/normed/group/basic.lean\#L55}
or ``$E$ is a normed group'' to state that
$(E,+,0)$ is an abelian group with a norm,
such that $d(x,y)\vcentcolon=\|x-y\|$ defines a metric on $E$.
A seminormed group\mllink{analysis/normed/group/basic.lean\#L49} has a similar definition,
with as only difference that $\|x\|=0$ need not imply that $x = 0$.

If $\kk$ is a normed field and $E$ is a normed group, then we have an additional type-class
\lstinline{[normed_space 𝕜 E]}\mllink{analysis/normed_space/basic.lean\#L39}
stating that $E$ is a vector space over $\kk$
with $\|k\cdot v\|=\|k\|\cdot\|v\|$ for $k\in\kk$ and $v\in E$.
The type-class \lstinline{normed_space} is a \emph{mixin},
which means that it takes other type-classes as arguments.
In this case the arguments are
\lstinline{[normed_field 𝕜]} and \lstinline{[seminormed_add_comm_group E]}.

The reason for this separation is that the class\\
\lstinline{[seminormed_add_comm_group E]}
contains all the properties that \emph{only} mention $E$. So for example,
if we want to use a lemma about the addition on $E$,
we only have to search for the instance \lstinline{[seminormed_add_comm_group E]} and not
for an instance of \lstinline{[normed_space 𝕜 E]}.
If we had a single type-class \lstinline{[normed_space 𝕜 E]} with all the properties,
then we would have to search for an instance \lstinline{[normed_space 𝕜 E]}
without knowing what $\kk$ is (since addition on $E$ does not refer to $\kk$).
With Lean's algorithm for type-class synthesis, it can be problematic to search for instances
without knowing all the arguments~\cite{baanen2022instances}.
In the definition given in \textsf{mathlib}, only properties of normed spaces that refer to
both $\kk$ and $E$ are in the \lstinline{[normed_space 𝕜 E]} class.
This means that when we use one of these properties,
both $\kk$ and $E$ are known once Lean has to synthesize an instance for this type-class.

Many particular properties of objects are also recorded in type-classes.
For example, the properties \lstinline{[sigma_finite μ]}, 
\lstinline{[is_add_left_invariant μ]} and \lstinline{[is_neg_invariant μ]}
(that we will define below) are all properties on measures that are marked as type-class arguments.
This means that we only have to prove these properties if they cannot be inferred from
the globally-declared instances, which does not come up often.
An example of such a global instance for these classes
states that $\mu\times\nu$ is $\upsigma$-finite
if both $\mu$ and $\nu$ are.\mllink{measure_theory/constructions/prod.lean\#L485}
\begin{lstlisting}
instance [sigma_finite μ] [sigma_finite ν] :
  sigma_finite (μ.prod ν)
\end{lstlisting}

\subsection{Filters}
\label{sub:filters}

\textsf{mathlib} has a general library for filters,
which are used throughout the topology, analysis and measure theory libraries.
Filters are used to make precise arguments that involve phrases such as
``for all $x$ close to $x_0$,'' ``for large enough $N$'' and ``for almost every $x\in \R$.''
In this way, filters can be used as a generalized version of bounded quantification.
The filter library is heavily inspired by the one in Isabelle~\cite{holzl2013analysis}.

Formally, a filter $F$ on $X$ is a collection of subsets of $X$ satisfying the following properties
\begin{itemize}
  \item $X\in F$;
  \item $F$ is upwards closed: if $X\in F$ and $Y\supseteq X$ then $Y \in F$;
  \item $F$ is closed under binary intersections: if $X,Y\in F$ then $X\cap Y\in F$.
\end{itemize}
Given a filter $F$, we can introduce the notion $\forallf{x}{F}, P(x)$
stating that $P$ holds eventually for $x \in F$. It is defined to mean that $\{x \mid P(x)\}\in F$.

Now the aforementioned examples become special cases of this notions.
``For all $x$ close to $x_0$, $P(x)$ holds'' means $\forallf{x}{\nhds_{x_0}}, P(x)$,
where the \emph{neighborhood filter} $\nhds_{x}$ is defined as
the collection of all neighborhoods of a point $x$ in a topological space:
$$\nhds_x=\{N \mid \text{there exists an open set $U\subseteq N$ with $x \in U$}\}.$$
Similarly, if $\mu$ is a measure on $X$ then ``for almost every $x\in X$, $P(x)$ holds'' means
$\forallf{x}{\text{AE}_\mu}, P(x)$, where the \emph{almost everywhere} filter of co-null sets is
defined to be
$$\text{AE}_\mu=\{A\mid \mu(X\setminus A)=0\}.$$
Note that the phrase ``$x$ is close to $x_0$'' by itself is not a precisely defined notion.
Filters only give a precise notion of ``for all $x$ close to $x_0$, $P(x)$ holds.''

We will see the definition of another filter in \Cref{sec:approximations}.

The filters $\mathcal{F}_X$ on $X$ are ordered by reverse inclusion:
$$F\le G\Longleftrightarrow G\subseteq F.$$

If $f:X\to Y$ then it induces a forward map $f_*:\mathcal{F}_X \to \mathcal{F}_Y$ defined by
$$A \in f_*(F) \Longleftrightarrow f^{-1}(A)\in F.$$
We can use filters to define limits and convergence.
If $f : X \to Y$, $F$ is a filter on $X$ and $G$ is a filter on $Y$,
we say that \emph{$f(x)$ tends to $G$ as $x$ tends to $F$} to mean that $f_*F\le G$.
One can derive that if $F$ and $G$ are neighborhood filters,
this exactly corresponds to the usual definition in a topological space.
If moreover $X$ and $Y$ are metric spaces,
it corresponds to the $\eps$-$\delta$ definition of limits.

\subsection{Differentiation}
\label{sub:differentiation}
Lean has a general theory of derivatives, formulated in terms of the Fréchet derivative
and developed mainly by \anon[{[redacted]}\footnote{In the anonymous version of this paper we avoid mentioning \textsf{mathlib} contributors by name.}]{Sébastien Gouëzel}. 

If $E$ and $F$ are normed spaces over a normed field $\kk$,
we denote the space of continuous $\kk$-linear maps from $E$ to $F$ as $L(E,F)=L_\kk(E,F)$.
The space $L(E,F)$ is itself a normed space,
where the norm $\|L\|$ is the least $c$ such that $\|Lx\|\le c\|x\|$ for all $x\in E.$
Suppose $f : E \to F$ is a map, $f' : L_\kk(E,F)$ is a continuous linear map and $x \in E$.
Then we write \lstinline{has_fderiv_at f f' x} for the statement that
$f$ has derivative $f'$ at $x$.
It is defined using the little-$o$ notation as
$f(x')-f(x)-f'(x'-x)=o(x'-x)$ as $x'$ tends to $x$.
If a derivative exists at $x$, we say that $f$ is differentiable at $x$.
We define the derivative \lstinline{fderiv 𝕜 f x} or $Df$ of $f$ at $x$
to be the derivative of $f$ at $x$ if it exists.
It is defined to be $0$ if $f$ has no derivative at $x$.

Furthermore, if $n \in \N\cup\{\infty\}$ then we write \lstinline{cont_diff 𝕜 n f}
or ``$f$ is $C^n$''
to state that $f$ is $n$-times continuously differentiable.
It is defined in terms of the existence of a Taylor series,
and it has the following defining properties.
\begin{lstlisting}
cont_diff 𝕜 0 f ↔ continuous f
cont_diff 𝕜 (n + 1) f ↔ differentiable 𝕜 f ∧
  cont_diff 𝕜 n (fderiv 𝕜 f)
cont_diff 𝕜 ∞ f ↔ ∀ (n : ℕ), cont_diff 𝕜 n f
\end{lstlisting}


\subsection{Integration}
\label{sub:integration}

Lean has a large library for measure theory and integration, including the Bochner integral,
Fubini's theorem, Differentiation under an integral
and the Haar measure~\cite{vandoorn2021haar}.\footnote{Results
in \textsf{mathlib} that we will not use in this paper include
the fundamental theorems of calculus,
change of variables for higher-dimensional integrals~\cite{gouezel2022changeofvariables},
and the divergence theorem~\cite{kudryashov2022integral}.}

Lean's Bochner integral is defined for maps from a measure space
into a Banach space (complete normed space) over $\R$.
Previously we required the codomain to be second countable,
but after a Herculean refactor by \anon[{[redacted]}]{Sébastien Gouëzel}%
\href{https://github.com/leanprover-community/mathlib/pull/12942}{\link}
this assumption is not necessary anymore.

A \emph{measurable space}\,\mllink{measure_theory/measurable_space_def.lean\#L50}
is a space equipped with a designated $\upsigma$-algebra of measurable sets,
and a \emph{measure space}\,\mllink{measure_theory/measure/measure_space_def.lean\#L444}
is a measurable space with a specified measure.
A \emph{simple function}\,\mllink{measure_theory/integral/lebesgue.lean\#L48}
is a function with finitely many values where all preimages are measurable.
If $f:X\to Y$ is a map from a measure space $(X,\mu)$ to a topological space $Y$,
then $f$ is \emph{$\mu$-a.e.~strongly measurable}%
\,\mllink{measure_theory/function/strongly_measurable.lean\#L103}
if it is $\mu$-almost everywhere equal to the limit of
a sequence of simple functions.
If the codomain of $f$ is a normed group, then $f$ is \emph{$\mu$-integrable} if
it is $\mu$-a.e.~strongly measurable and the Lower Lebesgue integral of $\|f\|$ is finite.\mllink{measure_theory/function/l1_space.lean\#L393}
$f$ is \emph{integrable on $A\subseteq X$}\,\mllink{measure_theory/integral/integrable_on.lean\#L80}
if $f$ is integrable w.r.t. to the measure $\mu|_A$ defined by $\mu|_A(B)=A\cap B$.
If $X$ is also a topological space, we say that $f$ is
\emph{locally integrable}\,\mllink{measure_theory/function/locally_integrable.lean\#L34} if
$f$ is integrable on all compact subsets of $X$.

One can define the Bochner integral by first defining it for simple $L^1$ functions
(as a simple linear combination),
and then extend it to all $L^1$ functions by noticing that
the simple $L^1$-functions are dense
in all $L^1$-functions.\mllink{measure_theory/function/simple_func_dense_lp.lean\#L686}.
Then the \emph{Bochner integral}\,\mllink{measure_theory/integral/bochner.lean\#L699}
$\int_X f(x)d\mu(x)$ of a function $f$
from a measure space $(X,\mu)$
to a Banach space $E$ is defined to be $0$ if $f$ is not integrable,
and otherwise it is defined to be the integral of $f$ viewed as an element of $L^1$.%

We say that a $\mu$ is
\emph{$\upsigma$-finite}\,\mllink{measure_theory/measure/measure_space.lean\#L2542} if there is a
countable collection of sets with finite measure that span $X$.
If $(X,\mu)$ and $(Y,\nu)$ are $\upsigma$-finite measure spaces,
then there is a \emph{product measure}\,\mllink{measure_theory/constructions/prod.lean\#L322}
$\mu\times\nu$ on $X\times Y$ satisfying
$$(\mu\times\nu)(A\times B)=\mu(A)\nu(B).$$

The Bochner integral satisfies Fubini's theorem~\cite{vandoorn2021haar}, which comes in two parts.

Fubini's theorem neatly expresses the integrability condition for functions $f:X\times Y\to E$.
If $f$ is $(\mu\times\nu)$-a.e.~strongly measurable then $f$ is $(\mu\times\nu)$-integrable iff
$y\mapsto f(x,y)$ is $\nu$-integrable for $\mu$-almost all $x\in X$ and
$x\mapsto\int_Y \|f(x,y)\|d\nu(y)$ is $\mu$-integrable.%
\mllink{measure_theory/constructions/prod.lean\#L866}
There is of course also a symmetric version of this condition.%
\mllink{measure_theory/constructions/prod.lean\#L874}

Furthermore, Fubini's theorem states how to compute the integral.
If $f:X\times Y\to E$ is $(\mu\times\nu)$-integrable, then%
\mllink{measure_theory/constructions/prod.lean\#L1025}
\begin{align*}
  \int_{X\times Y}f(z)d(\mu\times\nu)(z)&=\int_X\int_Y f(x,y)d\nu(y)d\mu(x)\\
  &=\int_Y\int_X f(x,y)d\mu(x)d\nu(y).
\end{align*}

Using the dominated convergence theorem, we can prove
under what conditions a parametric integral is continuous.
\begin{proposition}
\label{prop:int_continuous}
Let $Z$ be a first-countable topological space, $F: Z\times X\to E$ be a map and $z_0 \in Z$
such that
\begin{itemize}
  \item $F(z,{-})$ is $\mu$-a.e.~strongly measurable for $z$ near $z_0$;
  \item $F({-},x)$ is continuous for $\mu$-a.e. $x \in X$;
  \item There is an integrable bounding function $B:X\to\R$ such that
    $\|F(z,x)\|\le B(x)$ for $z$ near $z_0$ and $\mu$-almost all $x$.
\end{itemize}
Then $z\mapsto \int F(z,x)d\mu(x)$ is continuous at $z_0$.%
\mllink{measure_theory/integral/bochner.lean\#L898}
\end{proposition}

\anon[{[redacted]}]{Patrick Massot} has proven very general results in Lean that show that we can differentiate
under a parametric integral.
The result we will use is the following, but it is not the most general version that is provided.
\begin{proposition}
\label{prop:int_deriv}
Let $\kk=\R$ or $\kk=\C$, let $H$ be a normed space and $E$ be a Banach space, both over $\kk$.
Let $F: H\times X\to E$, $F':H\times X\to L_\kk(H,E)$, and $z_0\in Z$ such that
\begin{itemize}
  \item $F(z,{-})$ is $\mu$-a.e.~strongly measurable for $z$ near $z_0$;
  \item $F(z_0,{-})$ is $\mu$-integrable;
  \item $F'(z_0,{-})$ is $\mu$-a.e.~strongly measurable;
  \item $F({-},x)$ has derivative $F'(z,x)$ at $z$ for $\mu$-a.e. $x \in X$ and all $z$ near $z_0$;
  \item There is an integrable bounding function $B:X\to\R$ such that
    $\|F'(z,x)\|\le B(x)$ for all $z$ and $\mu$-almost all $x$.
\end{itemize}
Then $z\mapsto \int F(z,x)d\mu(x)$ has derivative $\int F'(z_0,x)d\mu(x)$ at $z_0$.%
\mllink{analysis/calculus/parametric_integral.lean\#L180}
\end{proposition}

Let $(G,+,0)$ be a measurable group.
This means that $G$ is a group equipped with a $\upsigma$-algebra such that
the addition ${+}:G\to G\to G$ and negation ${-}:G\to G$ are measurable.
If $\mu$ a measure on $G$,
we say that $\mu$ is \emph{invariant under left addition}
or \emph{left-invariant}\,\mllink{measure_theory/group/measure.lean\#L38} if
$\mu(g+A)=\mu(A)$ for all $g\in G$ and measurable $A\subseteq G$.
In this case, we also have that for any function $f:G\to E$
that\mllink{measure_theory/group/integration.lean\#L79}
$$\int_G f(g+x)d\mu(x)=\int_G f(x)d\mu(x).$$
Note that this equality also holds in $f$ is not integrable
and even if $f$ is not (a.e.) measurable. The reason is that $f$ is integrable if and only if
$x \mapsto f(g+x)$ is integrable, and if this is false, both sides of the equation equal $0$.

Similarly, we say that $\mu$ is \emph{invariant under negation}
or \emph{negation-invariant}\,\mllink{measure_theory/group/measure.lean\#L208}
if $\mu(-A)=\mu(A)$ for all measurable $A\subseteq G$.

If $G$ is a second-countable locally compact Hausdorff topological group,
then there is a left-invariant measure $\mu$ that is finite on compact sets,
and nonzero on non-empty open sets,
called the \emph{Haar measure}.\mllink{measure_theory/measure/haar.lean\#L565}
Moreover, this measure is
unique up to a scalar factor~\cite{vandoorn2021haar}.\mllink{measure_theory/measure/haar.lean\#L594}

\section{Definition of the convolution}%
\label{sec:convolution}

As mentioned in the introduction, we want to define a general definition of the convolution,
so that it is applicable in many situations.
For the convolution $f*g$ we don't want to require that the domains of $f$ and $g$ is $\R$,
or even any Euclidean space.
We can be more general by requiring that the domain $f$ and $g$ is an arbitrary group $(G,+,0)$
equipped with a $\upsigma$-algebra and a measure $\mu$.
In many of the properties we prove we will assume that
addition and negation are measurable operations, and that the measure is invariant under
addition and negation.
This is fine, since in most applications $G$ will be a locally compact group
with Haar measure $\mu$.
In fact, for the bare definition of convolution we don't even assume that $G$ is a group,
but just a type equipped with a binary subtraction operation ${-}:G\to G\to G$.

For the codomains of $f$ and $g$ we also want to have a general choice.
One of the applications we want to provide is when one starts with a function $f:\R^n\to\R^m$
that is merely continuous or even locally integrable.
One can take the convolution with a small bump function $\varphi:\R^n\to\R$ around the origin
to obtain a smooth function $\varphi*f$ that is close to the original function $f$.

This suggests the definition where we define the convolution $\varphi*_\mu f$ for
$\varphi : G \to \kk$ and $f: G \to E$ where $\kk$ is a normed field
and $E$ is a normed vector space over $\kk$, which we can define as
$$(\varphi *_\mu f)(x) = \int_G \varphi(t)f(x-t)d\mu(t).$$

This definition has two disadvantages.
First, it is asymmetric: the first function is scalar-valued
and the second function can be vector-valued.
This is not a problem per se, but it makes it awkward to work with.
We cannot state commutativity of the convolution in general,
but only in the special case where the second function is also scalar-valued.
Since we cannot state the commutativity of the convolution in general,
we cannot use it to prove symmetric versions lemmas.
Suppose we have proven that $\varphi *_\mu f$ is continuous
when $\varphi$ is continuous with compact support and $f$ is locally integrable.
We cannot conclude the symmetric version from this,
which states that $\varphi *_\mu f$ is also continuous when
$f$ is continuous with compact support and $\varphi$ is locally integrable.

Second, there is a problem with computing the total derivative of the convolution.
To work with derivatives, we have to assume that $G$ is also a normed vector space over $\kk$.
If $f$ is $C^1$ with compact support and $\varphi$ is locally integrable,
then we can prove that $D_x(\varphi*_\mu f)=(\varphi*_\mu Df)(x)$ for $x\in G$.
Note that on the right hand-side $Df:G\to L(G,E)$ is a family of continuous linear maps,
and so we take the convolution in the space $L(G,E)$ of continuous $\kk$-linear maps
from $G$ to $E$.

However, more commmonly we will be dealing with the symmetric case,
where $\varphi$ is $C^1$ with compact support and $f$ is locally integrable.
This case occur for example when we want to smoothen a vector-valued function
by taking the convolution with a smooth scalar-valued function.
In this case stating any formula for $D_x(\varphi*_\mu f)$ is difficult.
We want to say that it is the convolution of $D\varphi$ with $f$, evaluated at $x$,
but this does not fit in our current definition, since $D\varphi$ has type $L(G,\kk)$
and is not a scalar-valued function anymore. In this particular case,
it is not particularly hard to write down what the derivative should be, namely for $x,x'\in G$
\begin{equation}
  D_x(\varphi *_\mu f)(x') = \int_G D_t\varphi(x')f(x-t)d\mu(t).
  \label{eq:conv_diff_scalar}
\end{equation}

The right-hand side is almost a convolution, except that we have replaced scalar multiplication
inside the integral with a different bilinear map.
Indeed, we can define a bilinear map $A:L(G,\kk)\times E \to L(G,E)$ by $A(v,y)(x)=v(x)y$.
Now we can rewrite \eqref{eq:conv_diff_scalar} as
$$D_x(\varphi *_\mu f) = \int_G A(D_t\varphi,f(x-t))d\mu(t).$$
This suggests a more general definition of convolution that depends on a bilinear map.
This is the definition of convolution that we have formalized.

For the remainder of the paper,
we assume that $E$, $E'$ are normed spaces and $F$ is a Banach space over a normed field $\kk$.
Furthormore, we assume that $(G,\mu)$ is a measure space with a designated subtraction operation,
$L:E \times E' \to F$ is a continuous bilinear map,
and that $f:G\to E$ and $g:G\to E'$ are functions.

\begin{definition}
  The \emph{convolution of $f$ and $g$ with respect to $L$ and $\mu$}\,\clink{355} is defined by
  \begin{equation}
    (f*_{L,\mu}g)(x)\vcentcolon= \int_G L(f(t),g(x-t))d\mu(t).\label{eq:convolution_def}
  \end{equation}
  We say that \emph{the convolution $f*_{L,\mu}g$ exists at $x$}\,\clink{130}
  if the integrand of \eqref{eq:convolution_def} is integrable.
\end{definition}
Recall that an integral is defined to be $0$ if the integrand is not integrable.
Of course, for most lemmas about convolution we will need to assume that the convolution exists,
or a stronger condition that implies that the convolution exists.

If $G$ is nonabelian, there are multiple definitions of convolution.
We could have also defined the convolution $(f*_{L,\mu}g)(x)$ to be
$$\int_G L(f(-t+x),g(t))d\mu(t),$$
which is different unless $G$ is unimodular.
While we do not go into the theory of unimodular groups and the modular function
in the formalization, we generally avoid assuming that the group $G$ is abelian.
In this paper, we only assume that $G$ is abelian in the commutativity of the convolution,
and when we assume that $G$ is a (semi)normed group
in parts of Sections \ref{sec:smoothness} and \ref{sec:approximations}.

\section{Algebraic Properties}
\label{sec:algebra}

\begin{proposition}
The convolution satisfies the following algebraic properties for $f,f':G\to E$, $g,g':G\to E'$
and $c \in \kk$
\begin{enumerate}
\item $(cf)*_{L,\mu}g=f*_{L,\mu}(cg)=c(f*_{L,\mu}g)$;\,\clink{382}
\item if $f*_{L,\mu}g$ and $f*_{L,\mu}g'$ exist, then\,\clink{401}
  $$f*_{L,\mu}(g+g')=f*_{L,\mu}g+f*_{L,\mu}g';$$
\item if $f*_{L,\mu}g$ and $f'*_{L,\mu}g$ exist, then\,\clink{405}
  $$(f+f')*_{L,\mu}g=f*_{L,\mu}g+f'*_{L,\mu}g;$$
\item Let $L^t$ be the transpose of $L$, defined by $L^t(x,y)=L(y,x)$.
If $G$ is an abelian measurable group
and $\mu$ is invariant under left addition and negation,
then $f*_{L,\mu}g=g*_{L^t,\mu}f.$\clink{529}
\end{enumerate}
\end{proposition}
\begin{proof}
The first 3 items are immediate. The last equality follows from
\begin{align*}
  &\mathrel{\hphantom{=}} \int L(f(t),g(x-t))d\mu(t)\\
  &=\int L(f(x-s),g(x - (x - s)))d\mu(s)\\
  &=\int L^t(g(s),f(x-s))d\mu(s),
\end{align*}
where we used the invariance of $\mu$ in the first equality.
\end{proof}
Stating associativity is trickier in this general setting.
\begin{proposition}
Assume that $\kk$ is either $\R$ or $\C$ and
suppose that we are given a group $G$ with measurable addition,
$x_0\in G$,
a measure $\mu$ and a right-invariant measure $\nu$ on $G$, both of which are $\upsigma$-finite,
normed spaces $E_i$ and Banach spaces $F_i$ over $\kk$ ($i\in\{1,2,3\}$),
functions $f_i:G\to E_i$
and continuous bilinear maps
\begin{align*}
  L_1&:E_1\times E_2 \to F_1&L_3&:E_1\times F_2 \to F_3\\
  L_2&:F_1\times E_3 \to F_3&L_4&:E_2\times E_3 \to F_2.
\end{align*}
Suppose that
\begin{itemize}
  \item For all $x_i\in E_i$ we have
  \begin{equation}L_2(L_1(x_1,x_2),x_3)=L_3(x_1,L_4(x_2,x_3));\label{eq:assoc_assum}\end{equation}
  \item $f_1*_{L_4,\mu}f_2$ and $f_2*_{L_1,\nu}f_3$ exist;
  \item the map $(t,s)\mapsto L_3(f_1(s),L_4(f_2(t-s),f_3(x_0-t)))$ is $(\nu\times\mu)$-integrable,
\end{itemize}
then~\clink{804}
$$(f_1*_{L_1,\mu}f_2)*_{L_2,\nu}f_3 = f_1*_{L_3,\mu}(f_2*_{L_4,\nu}f_3).$$
\end{proposition}
Note that the conditions are immediately satisfied
if all the $L_j$ are (scalar) multiplication,
and all the $f_i$ are integrable.
\begin{proof}
We compute
\begin{align*}
  &\mathrel{\hphantom{=}} \int L_2\big({\textstyle\int} L_1(f_1(s),f_2(t-s))d\mu(s),f_3(x_0-t)\big)d\nu(t)\\
  &= \iint L_2(L_1(f_1(s),f_2(t-s)),f_3(x_0-t))d\mu(s)d\nu(t) \\
  &= \iint L_3(f_1(s),L_4(f_2(t-s),f_3(x_0-t)))d\mu(s)d\nu(t) \\
  &= \iint L_3(f_1(s),L_4(f_2(t-s),f_3(x_0-t)))d\nu(t)d\mu(s) \\
  &= \int L_3\big(f_1(s),{\textstyle\int}L_4(f_2(t-s),f_3(x_0-t)d\nu(t))\big)d\mu(s)\\
  &= \int L_3\big(f_1(s),{\textstyle\int}L_4(f_2(t),f_3(x_0-s-t)d\nu(t))\big)d\mu(s).
\end{align*}
In the first and fourth equality we use linearity (of $L_2$ resp. $L_3$),
which requires that the appropriate convolutions exist.
In the second equality we use \eqref{eq:assoc_assum},
in the third equality we use Fubini's theorem (using the integrability condition),
and in the fifth equality we use that $\nu$ is right-invariant.
\end{proof}

\section{Integrability}
\label{sec:integrability}

If $f$ and $g$ are both integrable, we cannot conclude that
the convolution $f*_{L,\mu}g$ exists everywhere. However, we can get close.

\begin{proposition}
If $G$ is a measurable group, $\mu$ is $\upsigma$-finite and right-invariant,
and both $f$ and $g$ are integrable, then
\begin{itemize}
  \item The convolution $f*_{L,\mu}g$ exists at $\mu$-almost all $x$ in $G$;\clink{246}
  \item $f*_{L,\mu}g$ is integrable.\clink{443}
  \item Furthermore, if $\kk$ is either $\R$ or $\C$ and both $E$ and $E'$ are Banach spaces, then%
  \clink{787}
    $$\int (f*_{L,\mu}g)(x)d\mu(x)=
    L\left({\textstyle\int} f(x)d\mu(x),{\textstyle\int} g(x)d\mu(x)\right).$$
\end{itemize}
\end{proposition}
\begin{proof}
We first prove that the map $(x,t)\mapsto L(f(t),g(x-t))$ is $(\mu\times\mu)$-integrable.
First note that this map is $(\mu\times\mu)$-a.e.~strongly measurable.
By Fubini's theorem, it is sufficient
to show that $$x\mapsto L(f(t),g(x-t))$$ is $\mu$-integrable for $\mu$-almost all $t$ and
that $$t\mapsto\int L(f(t),g(x-t))d\mu(x)$$ is $\mu$-integrable.

The first claim follows immediately from the integrability of $g$.
For the second claim, notice that
\begin{align*}
  &\mathrel{\hphantom{=}}\left\|\int L(f(t),g(x-t))d\mu(x)\right\|\\
  &\le \int \|L(f(t),g(x-t))\|d\mu(x)\\
  &\le \|L\|\cdot\|f(t)\| \cdot \int\! \|g(x-t)\|d\mu(x)\\
  &= \|L\|\cdot\|f(t)\| \cdot \int\! \|g(x)\|d\mu(x).
\end{align*}
We use the right-invariance of $\mu$ in the last equality.
This last expression is $\mu$-integrable in $t$ since $f$ is integrable, and therefore
$t\mapsto\int L(f(t),g(x-t))d\mu(x)$ is also $\mu$-integrable.

We can now apply the symmetric version of Fubini's theorem to conclude that
$$t\mapsto L(f(t),g(x-t))$$ is $\mu$-integrable for $\mu$-almost all $x$ and
that $$x\mapsto\int L(f(t),g(x-t))d\mu(t)$$ is $\mu$-integrable.

Finally, we can use Fubini's theorem one more time and use the bilinearity of $L$ twice
to calculate
\begin{align*}
  &\mathrel{\hphantom{=}}\iint L(f(t),g(x-t))d\mu(t)d\mu(x)\\
  &=\iint L(f(t),g(x-t))d\mu(x)d\mu(t)\\
  &=\int L\left(f(t),{\textstyle\int} g(x-t)d\mu(x)\right)d\mu(t)\\
  &=\int L\left(f(t),{\textstyle\int} g(x)d\mu(x)\right)d\mu(t)\\
  &=L\left({\textstyle\int} f(t)d\mu(t),{\textstyle\int} g(x)d\mu(x)\right).\qedhere
\end{align*}
\end{proof}

\section{Smoothness}%
\label{sec:smoothness}

One main property of the convolution is that it can smoothen a function.
That is to say, if one takes the convolution with a smooth function,
the convolution is smooth than the original function, even if the original function is
only locally integrable.
We will explore this behavior in this section.
In this section we will assume that $g$ is continuous/differentiable/$C^n$.
Using the commutativity of convolution one can show the same for the case where $f$ is smooth
(assuming that $G$ is abelian and $\mu$ left-invariant and negation invariant).

\subsection{Continuity and Differentiability}
\label{sub:continuity_diff}
In this paper, we denote the \emph{support}\,\mllink{algebra/support.lean\#L30} of a function $f$ as
$$\supp(f):=\{ x \mid f(x) \ne 0\}.$$
This notion should be contrasted
with the \emph{topological support}\,\mllink{topology/support.lean\#L44} of a function $f$,
which is the closure of $\supp(f)$:
$$\tsupp(f):=\overline{\supp(f)}.$$

We say that a function \emph{has compact support}\,\mllink{topology/support.lean\#L111} if its
topological support is compact.
\begin{proposition}
  \label{prop:continuous_conv}
  Suppose that $G$ is a second-countable locally compact Hausdorff topological group.
  If $f$ is locally $\mu$-integrable and $g$ is continuous with compact support
  then $f*_{L,\mu}g$ is continuous.~\clink{462}
\end{proposition}
\begin{proof}
  To show that $f*_{L,\mu}g$ is continuous at $x_0$, we use \Cref{prop:int_continuous}.
  Clearly, $x\mapsto L(f(t),g(x-t))$ is continuous.

  To show that the map $t\mapsto L(f(t),g(x-t))$ is $\mu$-a.e.~strongly measurable,
  note that $f$ is $\mu$-a.e.~strongly measurable,
  since we can write $G$ as the union of a countable collection of compact sets, and $f$ is
  $\mu$-a.e.~strongly measurable on each of those sets, hence on all of $G$.
  Furthermore, $L$ and $t\mapsto g(x-t)$ are continuous, hence $\mu$-a.e.~strongly measurable, so
  the composition is also $\mu$-a.e.~strongly measurable.

  To provide the integrable bound,
  let $K$ be a compact neighborhood of $x_0$ and let $K'=-\tsupp(g)+K$, which is also compact.
  For $x\in K$ and any $t\in G$ we have
  \begin{align*}
    &\mathrel{\hphantom{=}} |L(f(t),g(x-t))\| \\
    &\le\|L\|\cdot\|f(t)\|\cdot\|g(x-t)\|\\
    &\le \begin{cases} \|L\|\cdot\|f(t)\|\cdot\sup\|g\| & \text{if $t \in K'$}\\
      0 &\text{otherwise} \end{cases}
  \end{align*}
  Since $f$ is integrable on $K'$, this last function is integrable.
\end{proof}

We have to be careful with the above statement:
our definition of integration includes the integral is defined to be $0$
if the integrand is not integrable. So we have to separately prove that the convolution
really exists. In this case, that is indeed true, since the integrand is a
$\mu$-a.e.~strongly measurable function that is bounded by an integrable function,
so it is itself integrable.

We also prove a variant of \Cref{prop:continuous_conv} that with a stronger condition on $f$,
namely that that $f$ is $\mu$-integrable, but weaker conditions on $g$, namely that
$g$ is continuous and $\|g\|$ is bounded.\clink{485}

For the remainder of this section, assume that $\kk=\R$ or $\kk=\C$,
that $G$ is a finite-dimensional normed space over $\kk$ and that
$\mu$ is $\upsigma$-finite and left-invariant.

As in the discussion in \Cref{sec:convolution}, we have to define a new bilinear map to state
what the derivative of a convolution is.
In this case, we define the continuous linear map $$L^G:E\times L(G, E') \to L(G,F)$$ as
$$L^G(x,f)=L(x,{-})\circ f.$$
We can now write down the derivative of a convolution elegantly.

\begin{proposition}
  \label{prop:derivative_conv}
  If $f$ is locally $\mu$-integrable and $g$ is $C^1$ with compact support
  then $f*_{L,\mu}g$ has derivative $f*_{L^G,\mu}Dg$.~\clink{841}
\end{proposition}
\begin{proof}
We use \Cref{prop:int_deriv} with $F(x,t):=L(f(t),g(x-t))$ and $F'(x,t):=L^G(f(t),D_{x-t}g)$.
We already showed that $F(x,{-})$ is integrable and hence $\mu$-a.e.~strongly measurable.
Since $Dg$ is continuous, $F'(x,{-})$ is also $\mu$-a.e.~strongly measurable.
The reason that the derivative is dominated by an integrable function
is the same as in the proof of \Cref{prop:continuous_conv}.
Finally, to compute the derivative of $F$ we use the chain rule:
\begin{align*}
  D_xF({-},t)&= D_{g(x-t)}L(f(t),{-})\circ D_x(x \mapsto g(x-t))\\
  &= L(f(t),{-})\circ D_{x-t}g\\
  &= L^G(f(t),D_{x-t}g) = F'(x,t).\qedhere
\end{align*}
\end{proof}

From this, we get the appropriate smoothness condition for convolutions.
\begin{proposition}
  \label{prop:Cn_conv}
  If $f$ is locally $\mu$-integrable and $g$ is $C^n$ with compact support
  (for $n \in \N\cup\{\infty\}$)
  then $f*_{L,\mu}g$ is $C^n$.~\clink{878}
\end{proposition}

\subsection{A Digression on Universes}
\label{sub:universes}
Before we delve into the proof of \Cref{prop:Cn_conv},
we need to discuss a foundational obstacle in Lean
when doing inductive proofs of this form.
The idea of the this proof is by performing induction on $n$.
In the successor case, one can use the equivalence
\begin{lstlisting}
cont_diff 𝕜 (n + 1) f ↔ differentiable 𝕜 f ∧
  cont_diff 𝕜 n (fderiv 𝕜 f)
\end{lstlisting}
To show that $f*_{L,\mu}g$ is differentiable we use \Cref{prop:derivative_conv}
and we use the induction principle to show that $D(f*_{L,\mu}g)$ is $C^n$.
The problem is that $f*_{L,\mu}g$ has type $G \to F$ and $D(f*_{L,\mu}g)$ has type $G \to L(G,F)$.
The types $G$ and $F$ live in certain universes, say $\U_a$ resp. $\U_b$
(this would be denoted \lstinline{Type a} and \lstinline{Type b} in Lean).
In \textsf{mathlib}, we want to be maximally general,
so we want to allow for the possibility that $G$ and $F$ live in different universes,
so we don't want to assume that $a$ and $b$ are the same.
Now the catch is that $L(G,F)$ has type $\U_{\max(a,b)}$.
Since Lean does not have cumulative universes~\cite{vandoorn2018thesis},
we cannot get $E$ to have type $\U_{\max(a,b)}$ automatically.
This is a problem for the induction:
the universe levels in the base case have to be the same as the universe levels
in the induction hypothesis.
A similar problem occurs with the derivative of $g:G\to E'$,
since $Dg$ has type $G \to L(G,E')$, which lives in a different universe than $E'$.

To fix the induction argument, one has to assume that the universe levels of $E'$ and $F$ are larger
than the universe level of $G$. One can do that by assuming that
$G:\U_a$, $E:\U_b$, $E':\U_{\max(a,c)}$ and $F:\U_{\max(a,d)}$. In this case, the type of $E'$ and
$L(G,E)$ are the same, namely $\U_{\max(a,c)}$.
Similarly, the types of $F$ and $L(G,F)$ are the same, and the induction goes through.
To get the fully general version of the lemma without universe restrictions, one can use
\emph{universe lifting} $\ulift:\U_a\to\U_{\max(a,b)}$.
$\ulift(F)$ is a copy of type $F$ in a higher universe level,
which comes with a bijection $\up:F\to\ulift F$.
We can apply the version of the lemma with restricted universes to $f*_{\up\circ L,\mu}(\up\circ g)$
to derive it for $f*_{L,\mu}g$. However, this is unpleasant to do, since it requires a lot of lemmas
stating that the map $\up$ commutes with various operations and preserves various properties (for example,
to write the aforementioned convolution, we need to know that $\up$ is a continuous linear map).

This proof works, but it is not convenient to give, especially since these induction proofs are very
common to prove that some operation respects $C^n$ function.
Note that with universe cumulativity this problem disappears: if $F:\U_a$ then $F$ also has type
$\U_{\max(a,b)}$.
Usually the lack of universe cumulativity is completely invisible to the user,
and occasionally it is a very minor inconvenience
(for example, when one has to write down universe levels explicitly).
This is the only issue that the author of this paper has encountered where the lack of
universe cumulativity is annoying.\footnote{But it should be noted that
there are also issues with universe cumulativity~\cite{luo2012universes}.}

However, for our current purpose there is a much simpler proof
that avoids all of these universe issues.

We have proven a different version of the characterization of being $C^{n+1}$
(in collaboration with \anon[{[redacted]}]{Patrick Massot} and \anon[{[redacted]}]{Oliver Nash}),
namely\mllink{analysis/calculus/cont_diff.lean\#L3109}
\begin{lstlisting}
cont_diff 𝕜 (n + 1) f ↔ differentiable 𝕜 f ∧
  ∀ v, cont_diff 𝕜 n (λ x, fderiv 𝕜 f x v)
\end{lstlisting}
However, we have proven this only in the case that the domain of $f$ is finite dimensional.
The difference with the previous formulation is that instead of stating that $Df$ is $C^n$,
viewed as a map into the space of continuous linear maps, we look at all directional derivatives,
and require that $x\mapsto D_xf(v)$ (denoted $D_{({-})}f(v)$) is smooth for all vectors $v$.

The reason that this avoids the universe issue mentioned above
is that the map $D_{({-})}f(v)$ has the same domain and codomain as $f$,
and therefore the induction works without any universe restrictions.

The proof of this characterization of being $C^{n+1}$ is not very difficult
if the domain is finite dimensional,
since we can reconstruct a linear map from its values on set of vectors forming a basis,
and this reconstructed linear map is $C^n$ iff its values on the basis are $C^n$.
However, this characterization of being $C^{n+1}$ is false when the domain of
$f$ is not finite-dimensional, so
we cannot use this property for all proofs.
In particular, the proof that the composition of $C^n$ functions is $C^n$
(done by \anon[{[redacted]}]{Sébastien Gouëzel})
still uses a cumbersome proof that has to work with universe lifts.

\subsection{Proving Smoothness}
\label{sub:smoothness}
In \Cref{prop:Cn_conv}, we are assuming that the domain $G$ of all functions involved is
finite dimensional, so we can use the the characterization of being $C^{n+1}$
described in \Cref{sub:universes}.

\begin{proof}[Proof of \Cref{prop:Cn_conv}]
We perform induction on $n$. For $n=0$ it follows directly from \Cref{prop:continuous_conv}.
If $n=k+1$, then we know that $g$ is differentiable,
so $f*_{L,\mu}g$ is differentiable by \Cref{prop:derivative_conv}.
We can also compute
\begin{align*}
  D_x(f*_{L,\mu}g)(v)&=(f*_{L^G,\mu}Dg)(x)(v)\\
  &=\left(\int L^G(f(t),D_{x-t}g)d\mu(t)\right)(v)\\
  &=\int L(f(t),D_{x-t}g(v))\circ d\mu(t)\\
  &=(f*_{L,\mu}D_{({-})}g(v))(x).
\end{align*}
Note that $D_{({-})}g(v)$ has compact support and is $C^n$, so
we can apply the induction hypothesis to $f$ and $D_{({-})}g(v)$.
We conclude that $f*_{L,\mu}g$ is $C^{n+1}$.

We automatically get the case $n=\infty$,
since being $C^\infty$ is equivalent to being $C^n$ for all $n$.
\end{proof}

\section{Approximations of Functions}%
\label{sec:approximations}

In this section we will discuss how the convolution can be used to approximate functions.
If $x$ is a point in a metric space, we write
$$\ball_R(x)=\{ y \mid d(x,y)<R\}$$
for the open ball around $x$ with radius $R$.

In this section, let $G$ be a second-countable seminormed group, $\mu$ a left-invariant
$\upsigma$-finite measure, $\kk=\R$ and $E$ be a Banach space.
\begin{lemma}
  \label{lem:conv_dist}
  Let $x_0\in G$ and assume that
  $f$ is $\mu$-integrable with $\supp(f)\subseteq \ball_R(0)$.
  Suppose that $g$ is strongly $\mu$-a.e.~measurable
  such that for all $x \in \ball_R(x_0)$ we have $d(g(x),g(x_0))\le \eps$ for some $\eps\ge0$.
  Then the distance between $(f*_{L,\mu}g)(x_0)$ and $\int_GL(f(t),g(x_0))d\mu(t)$
  is at most $$\eps\|L\|\cdot\int_G\!\|f(x)\|d\mu(x).\clink{610}$$
\end{lemma}
\begin{proof}
  First note that the convolution $f*_{L,\mu}g$ exists at $x_0$, since $f$ is integrable on
  $\ball_R(0)$ and $g$ is bounded on $\ball_R(x_0)$ (we omit the details).
  Note that for all $t\in G$ we have
  \begin{equation}
    d(L(f(t),g(x_0-t)),L(f(t),g(x_0)))\le \eps\|L(f(t),{-})\|.\label{eq:dist_ineq}
  \end{equation}
  Indeed, if $t\in\supp(f)\subseteq \ball_R(0)$ then $d(g(x),g(x_0))\le \eps$,
  hence \eqref{eq:dist_ineq} holds. Otherwise the inequality is also valid,
  since both sides are equal to $0$. Therefore,
  \begin{align*}
    &\mathrel{\hphantom{=}}
      d\left((f*_{L,\mu}g)(x_0), {\textstyle\int_G}L(f(t),g(x_0))d\mu(t)\right)\\
    &\le \int_G \eps\|L(f(t),{-})\|d\mu(t)\\
    &\le \int_G \eps\|L\|\cdot\|f(t)\|d\mu(t)\\
    &= \eps\|L\|\cdot\int_G\!\|f(t)\|d\mu(t)
  \end{align*}
  as desired.
\end{proof}

Write $\varphi*_{\mu}g$ for the convolution when the bilinear map is scalar multiplication.

Given a filter $M$ on $X$ we can define the filter ``small sets w.r.t. $M$''
on the power set $\mathcal{P}(X)$,
defined as the largest filter containing $\mathcal{P}(A)$
for $A\in M$.\mllink{order/filter/small_sets.lean\#L32}

The defining property is that if $f:Y\to \mathcal{P}(X)$ is a family of sets
and $N$ is a filter on $Y$, then
$f(x)$ tendsto the small sets w.r.t. $M$ as $x$ tends to $N$
if and only if for all $A\in M$ we have\mllink{order/filter/small_sets.lean\#L46}
$$\forallf{y}{N}, f(x)\subseteq A.$$

\begin{proposition}
  \label{prop:conv_tendsto}
  Let $M$ be a filter on some indexing type $I$
  and let $(\varphi_i)_{i\in I}$ be a sequence of nonnegative functions $G\to\R$ with integral 1.
  Suppose that $\supp(\varphi_i)$ tends to small sets w.r.t. $\nhds_0$ as $i$ tends to $M$.
  If $g$ is a strongly $\mu$-a.e.~measurable function continuous at $x_0\in G$, then
  $i\mapsto (\varphi_i*_{\mu}g)(x_0)$ tends to $\nhds_{g(x_0)}$ as $i$ tends to $M$.
  \clink{677}
\end{proposition}
\begin{proof}
We have to show for all $\eps>0$ that
$$\forallf{i}{M}, d((\varphi_i*_{\mu}g)(x_0),g(x_0))<\eps.$$
Since $g$ is continuous at $x_0$, we find a $\delta>0$ such that $d(g(x),g(x_0))<\frac\eps2$
for $x\in\ball_\delta(x_0)$.
By the condition on $\supp(\varphi_i)$ we know that
$$\forallf{i}{M}, \supp(\varphi_i)\subseteq\ball_\delta(0).$$
If $i\in I$ such that $\supp(\varphi_i)\subseteq\ball_\delta(0)$ then
it suffices to show that $d((\varphi_i*_{\mu}g)(x_0),g(x_0))<\eps$.
We apply \Cref{lem:conv_dist} (noting that $\varphi_i$ is integrable, since its integral is nonzero)
to conclude that
$$d\left((\varphi_i*_\mu g)(x_0), {\textstyle\int_G}\varphi_i(t)g(x_0)d\mu(t)\right)
\le\frac\eps2<\eps,$$
using that the norm of scalar multiplication is 1.
We obtain the desired result by noting that
\begin{equation*}
  \int_G\varphi_i(t)g(x_0)d\mu(t)=g(x_0)\int_G\varphi_i(t)d\mu(t)=g(x_0).\qedhere
\end{equation*}
\end{proof}

In the formalization the of the proof of \Cref{prop:conv_tendsto} is a combination of two lemmas.
The combined length of the proofs of these lemmas is 14 lines long,
showing that \textsf{mathlib} provides a useful library for topology and calculus,
where we can efficiently prove tricky arguments.

\Cref{prop:conv_tendsto} is particularly useful when combined with
\anon[{[redacted]}]{Yury Kudryashov}'s work on defining smooth bump functions $E\to\R$,%
\mllink{analysis/calculus/specific_functions.lean\#L481}
where $E$ is a finite dimensional normed vector space over $\R$.
Suppose $\varphi_R$ is such a smooth bump function on $G$ with $\tsupp(\varphi_R)=B_R(0)$,
normed so that $\int\varphi_R(x)d\mu(x)=1$.
Using these functions, we get a much simpler corollary from \Cref{prop:conv_tendsto}.

\begin{corollary}
  \label{cor:conv_tendsto_bump}
  If $g$ is a strongly $\mu$-a.e.~measurable function continuous at $x_0\in G$, then
  $R\mapsto (\varphi_R*_{\mu}g)(x_0)$ tends to $\nhds_{g(x_0)}$ as $R$ tends to $0$.
  \clink{750}
\end{corollary}

\section{Concluding Thoughts}%
\label{sec:conclusion}

This formalization was added to \textsf{mathlib} in a sequence of 12 pull requests.
These pull requests added roughly 400 lines of code to existing files
throughout the filter, analysis and measure theory libraries
and 900 lines of code to a new file about convolutions.

Therefore, this is not a particularly long formalization,
despite the various results we have formalized.
The shortness is a virtue:
it shows that one can do real mathematics and do complex arguments in analysis
with relatively little effort.

We notice that the type-classes Lean uses reduce the cognitive overhead of the formalizer.
Lean regularly gives an error, complaining that there is a missing type-class assumption in a
proposition statement. Furthermore, \textsf{mathlib} has a suite of linters
that check if lemmas unused type-class assumptions~\cite{vandoorn2020maintaining}.
It can also automatically check
whether a weaker type-class suffices for a particular lemma~\cite{best2021typeclasses}.
This means that the end-user can give a quick approximation of the needed type-classes,
and Lean will give feedback if more or fewer are needed.

This convolution library has already been used in the Sphere Eversion Project\href{https://leanprover-community.github.io/sphere-eversion/}{\link} to show that a particular sequence of loops
approximate a given loop using \Cref{cor:conv_tendsto_bump}.
\href{https://github.com/leanprover-community/sphere-eversion/blob/79d11867588fd6725e6babc1d7c9093d09bded26/src/loops/delta_mollifier.lean#L322}{\link}

It would be interesting to write more theories on top of this library for convolutions.
In particular it would be interesting to define distributions.
The inclusion map from locally integrable function to distributions can be defined in terms of a
convolution of two functions, and the properties of this map should be easily derivable
from the properties in this paper.

In Sections \ref{sec:smoothness} and \ref{sec:approximations} we focused on the convolution
when one of the functions is compactly supported.
There are other conditions one can impose on the functions to still get a nice theory of
convolutions.
For example, one can relax the condition that one of the function is compactly supported,
by requiring that both functions are rapidly decreasing.
It would be interesting to investigate the convolution with other requirements on the functions.

\section*{Acknowledgements}
I would like to thanks to Patrick Massot for the mentoring and collaboration during my postdoc,
and for proofreading the first version of this article.
I would like to thank everyone in the \textsf{mathlib} community for developing a usable and convenient library.
Finally, I would like to thank Fondation Mathématique Jacques Hadamard for the financial support for this project.

\Urlmuskip=0mu plus 1mu 
\bibliographystyle{amsalphaurl}
\bibliography{references}
\end{document}

%% file: macros_env.tex
\newcommand{\nhds}{\mathcal{N}}

\def\U{\mathcal{U}}

\newcommand{\N}{\mathbb{N}}

\newcommand{\Z}{\mathbb{Z}}

\newcommand{\R}{\mathbb{R}}

\def\C{\mathbb{C}}
\newcommand{\anon}[2][1]{#1}

\newcommand{\forallf}[2]{\forall^f #1 ∈ #2}

\DeclareMathOperator{\supp}{supp}

\DeclareMathOperator{\tsupp}{tsupp}

\DeclareMathOperator{\ulift}{ulift}

\DeclareMathOperator{\up}{up}

\DeclareMathOperator{\ball}{B}

\newcommand{\eps}{\varepsilon}

\newcommand{\link}{\ensuremath{{}^\text{\faExternalLink}}}

\newcommand{\kk}{\Bbbk}